\documentclass[reprint, aps, pra, twocolumn, superscriptaddress]{revtex4-1}
\usepackage{graphicx}
\usepackage{dcolumn}
\usepackage{bm}
\usepackage{amssymb}
\usepackage{amsmath}
\usepackage{amsthm}
\usepackage{color}

\usepackage{hyperref}
\hypersetup{
	colorlinks=true,
	linkcolor=blue,
	filecolor=magenta,      
	urlcolor=cyan,
	pdftitle={Overleaf Example},
	pdfpagemode=FullScreen,
}

\usepackage{braket}
\usepackage{physics}

\newcommand\inner[2]{\left\langle #1, #2 \right\rangle}

\newtheorem{theorem}{Theorem}[]
\newtheorem{definition}[theorem]{Definition}

\newtheorem{lemma}[theorem]{Lemma}
\newtheorem{proposition}[theorem]{Proposition}

\begin{document}


\title{SDP for One-shot Dilution of Quantum Coherence}

\author{Yikang Zhu}
\affiliation{School of Computer Science and Technology, University of Science and Technology of China, Hefei 230026, China.} 

\author{Zhaofeng Su}
\email{zfsu@ustc.edu.cn}
\affiliation{School of Computer Science and Technology, University of Science and Technology of China, Hefei 230026, China.} 
\affiliation{CAS Key Laboratory of Wireless-Optical Communications, University of Science and Technology of China, Hefei 230026, China.}


\date{\today}

\begin{abstract}
   Quantum coherence is one of the fundamental properties of quantum mechanics and also acts as a valuable resource for a variety of practical applications, which includes quantum computing and quantum information processing.
   Evaluating the dilution of coherence is a basic problem in the framework of resource theory. We consider the coherence dilution problem in the one-shot scenario.
   We find a semidefinite program of one-shot coherence dilution of pure state under maximally incoherent operation. We further give a similar but not semidefinite program form under dephasing-covariant incoherent operation.
   Moreover, we prove that the known lower bound of the one-shot dilution is strict.
   Our numerical experiment clearly demonstrates that the maximally incoherent operation and dephasing-covariant incoherent operation have different power in the coherence dilution.
  
\end{abstract}

\maketitle

\section{Introduction}

Coherence is a fundamental property of quantum systems, arising from the phenomenon of quantum superposition. Quantum coherence plays an indispensable role in a variety of fields, for example, quantum thermodynamics \cite{PhysRevLett.111.250404}, quantum biology~\cite{sciadv.aaz4888}, quantum metrology \cite{giorda2017coherence}, quantum information processing protocols such as quantum channel discrimination \cite{PhysRevLett.116.150502, PhysRevA.93.042107} and quantum state merging~\cite{PhysRevLett.116.240405}, quantum computing tasks such as quantum algorithms \cite{PhysRevA.93.012111, PhysRevA.95.032307}. In recently years, the phenomenon of quantum coherence has been realized in quantum computing engineering~\cite{siddiqi2021engineering, luo2023quantum}. 

Due to the development of quantum information science, quantum coherence acts as a valuable resource which is under the framework of resource theory.
A general resource theory considers two fundamental ingredients, namely a free set of states and a free set of operations. The definition of free operation varies on motivations, while the free set of states is closed under any free operation in the corresponding free set of operations. States that are not in the free set are called resourceful states~\cite{aberg2006quantifying, baumgratzQuantifyingCoherence2014}.
Quantifying the resource is a fundamental question in the resource theory. Several measures have been proposed for the quantification problem, which include distance-based constructions, entropic measures, geometric measures, and witness-based measures~\cite{chitambarQuantumResourceTheories2019}. 

In the resource theory of coherence, free states are quantum states that are diagonal in a referenced computational basis, which are also known as incoherent states.
There have been proposed four important classes of free operations, the maximally incoherent operation (MIO) \cite{aberg2006quantifying}, the dephasing-covariant incoherent operation (DIO) \cite{PhysRevLett.117.030401, PhysRevA.94.052324}, the incoherent operation (IO) \cite{baumgratzQuantifyingCoherence2014} and the strictly incoherent operation (SIO) \cite{winterOperationalResourceTheory2016}. 
The most frequently used measures of quantum coherence include the relative entropy of coherence $C_r$ \cite{baumgratzQuantifyingCoherence2014} and the coherence of formation $C_f$ \cite{yuanIntrinsicRandomnessMeasure2015}.

The transformation of quantum states under free operations is a research topic of greatest interest in the field of resource theory~\cite{chitambarQuantumResourceTheories2019}. The procedure for converting a given quantum state $\rho$ into the canonical resource state $\Psi_d$ is known as \textit{distillation} while the reverse procedure is known as \textit{dilution}. Consider the transformation $\rho^{\otimes n} \rightarrow \Psi_2^{\otimes n R}$ under free operations. Suppose infinite independent and identical copies of $\rho$ can be used in the scenario. The maximal proportion $R$ is defined as the asymptotic distillation rate of the state $\rho$. The similar definition can be applied for the asymptotic dilution rate. The conversions between pure states under IO and SIO have been extensively studied ~\cite{winterOperationalResourceTheory2016, duConditionsCoherenceTransformations2015, zhuOperationalOnetooneMapping2017}. Winter and Yang showed that the asymptotic rate of distillation and dilution under IO and SIO are the corresponding relative entropy of coherence $C_r$ and the coherence of formation $C_f$, respectively \cite{winterOperationalResourceTheory2016}.

In practical terms, the resources are finite and the number of quantum states prepared for information processing tasks is also limited. Therefore, it is necessary to consider the one-shot scenario and investigate the corresponding distillation and dilution rates of coherence. Bartosz and his collaborators developed the framework of coherence distillation under one-shot scenario~\cite{regulaOneShotCoherenceDistillation2018}. They proposed a semidefinite program (SDP) for efficiently computing the one-shot distillation of coherence and showed that MIO and DIO have the same power in the scenario. The framework for one-shot scenario coherence dilution was proposed nearly at the same time and the relation between coherence dilution under MIO, DIO, IO and SIO were systematical investigated in perspective of inequalities~\cite{zhaoOneShotCoherenceDilution2018}.  
Hayashi and his collaborators obtained the second order asymptotics of coherence distillation, which is a more accurate approximation~\cite{hayashiFiniteBlockLength2021}.
To date, the second order of dilution remains an open problem.

In this paper, we investigate the one-shot scenario of coherence dilution. 
We find a SDP of one-shot coherence dilution of pure state under MIO and give a similar but not SDP form under DIO. 
Based on our result, accurate value of $C_{\mathcal{O}}^{(1), \epsilon}$ can be obtained through numerical calculations.
For the case of pure states, we prove that our results of both MIO and DIO coincide with the lower bounds proposed by Zhao, \textit{et al.}\cite{zhaoOneShotCoherenceDilution2018}.
Through numerical experiments, we find that there is a gap between $C_{\text{MIO}}^{(1), \epsilon}$ and $C_{\text{DIO}}^{(1), \epsilon}$. Although the power of MIO and DIO are the same for the dilution in the asymptotic scenario and the coherence distillation, our result demonstrates that MIO outperforms than DIO in the coherence dilution. 

\section{Preliminary}

Coherence is a basis-dependent concept, which means we must fix the referenced orthonormal basis. 
Given a $d$-dimensional Hilbert space $\mathcal{H}$, we use $\mathcal{B}(\mathcal{H})$ to denote the set of all bounded trace class operators on $\mathcal{H}$, and $\mathbb{D} = \{ \rho \in \mathcal{B}(\mathcal{H}) \mid \tr \rho = 1, \rho \geq 0 \}$ to denote all density operators. We choose an orthonormal basis $\{ \ket{i} \mid i = 0, \ldots, d-1 \}$, the density operators that are diagonal in this referenced basis form the incoherent set $\mathcal{I} \subset \mathbb{D}$. 
Therefore, all incoherent states $\rho \in \mathcal{I}$ are of the form
\begin{equation}
	\rho = \sum_{i=0}^{d-1} p_i \ketbra{i}{i}
\end{equation}
with probabilities $p_i$. If $\rho \notin \mathcal{I}$, we say $\rho$ is coherent or $\rho$ has coherence.

The definition of free operations for the resource theory of coherence is not unique. Here we list some important classes. 
The largest class is the maximally incoherent operations (MIO) \cite{aberg2006quantifying}, which contains completely positive trace preserving (CPTP) maps $\Lambda$ such that $\Lambda(\delta) \in \mathcal{I}$, $\forall\, \delta \in \mathcal{I}$. MIO can't produce coherence from incoherent states.
An interesting subset of MIO are the dephasing-covariant incoherent operations (DIO) \cite{PhysRevLett.117.030401, PhysRevA.94.052324}, which contain CPTP maps $\Lambda$ such that $\Lambda \circ \Delta = \Delta \circ \Lambda$, where $\Delta$ is the completely dephasing channel $\Delta(\cdot) := \sum_{i=0}^{d-1} \ketbra{i} \cdot \ketbra{i}$. An equivalent definition of DIO is $\Delta \circ \Lambda(\ketbra{i}{j}) = 0$ for $i \neq j$.
Another subset of MIO are the incoherent operations (IO) \cite{baumgratzQuantifyingCoherence2014}, which are CPTP maps that admit a Kraus operator representation $\Lambda(\rho) = \sum_n K_n \rho K_n^\dagger$ with $\{K_n\}$ being incoherent-preserving operators, that says $K_n \rho K_n^\dagger = p \delta$, where $0 \leq p \leq 1$ and $\delta \in \mathcal{I}$ for any $\rho \in \mathcal{I}$.
Finally, the strictly incoherent operations (SIO) \cite{winterOperationalResourceTheory2016} are CPTP maps that admit a Kraus operator representation $\Lambda(\rho) = \sum_n K_n \rho K_n^\dagger$ with both $\{K_n\}$ and $\{K_n^\dagger\}$ being incoherent-preserving operators.
Fig.~\ref{fig:operations} shows the relationship between the above four free operations.

\begin{figure}[htb]
	\centering
	\includegraphics[width=0.3\textwidth]{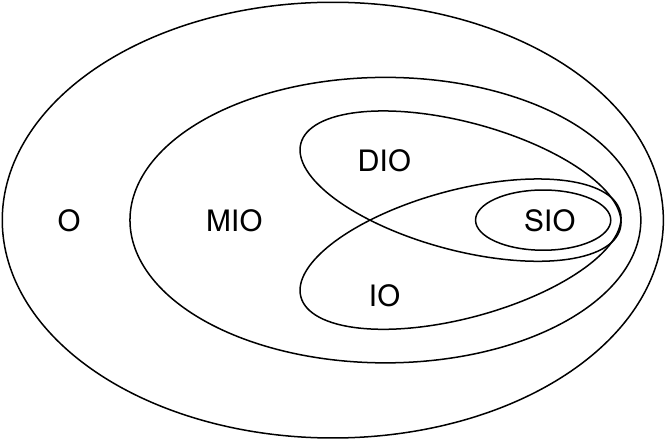}
	\caption{Relationship between MIO, DIO, IO and SIO.}
	\label{fig:operations}
\end{figure}

The canonical maximally coherent state of dimension $d$ is defined as 
\begin{equation}
	\ket{\Psi_d} := \frac{1}{\sqrt{d}} \sum_{i=0}^{d-1} \ket{i} 
\end{equation}
and we often denote $\Psi_d := \ketbra{\Psi_d}$. In the resource theory, we usually consider how to transform general states to the maximally resourceful state using free operations, and vice versa. These two kinds of transformation are called distillation and dilution. Formally, we can define the one-shot coherence dilution as the following.

\begin{definition}[One-shot coherence dilution]
	Let $\mathcal{O} \in \{\text{MIO}, \text{DIO}, \text{IO}, \text{SIO}\}$ denote a class of free operations.
	Given the state $\rho \in \mathbb{D}$ and $\epsilon \geq 0$, the one-shot coherence dilution rate (also called coherence cost) is defined as
	\begin{equation}
		C_{\mathcal{O}}^{(1), \epsilon}(\rho) := \min\{ \log_2 m \mid \Lambda \in \mathcal{O}, F[\Lambda(\Psi_m), \rho] \geq 1- \epsilon \} \label{def:dilution}
	\end{equation}
\end{definition}

The dilution procedure concerns how much coherence will be consumed to produce the given state $\rho$ using the free operations. Here we use $\Psi_2$ as the unit of coherence resource, so $\Psi_m$ has $\log_2 m$ units of coherence. In this definition, we use the fidelity $F(\rho, \sigma) := (\tr [\sqrt{\sqrt{\rho} \sigma \sqrt{\rho}}] )^2$ to measure the closeness between two states.

The superoperator is a linear map $\mathcal{E} : \mathcal{B}(\mathcal{H}) \rightarrow \mathcal{B}(\mathcal{H'})$, that means for any $\rho, \sigma \in \mathcal{B}(\mathcal{H})$ and any $\lambda, \mu \in \mathbb{C}$, we have
\begin{equation}
	\mathcal{E}(\lambda \rho + \mu \sigma) = \lambda \mathcal{E}(\rho) + \mu \mathcal{E} (\sigma)
\end{equation}
The inner product between two operators is defined as $\inner{X}{Y} := \tr X^\dagger Y$, and the adjoint of a superoperator $\mathcal{E}$ is defined as $\inner{X}{\mathcal{E}(Y)} = \inner{\mathcal{E}^\dagger (X)}{Y}$. A superoperator $\mathcal{E}$ is called trace preserving (TP) if $\tr [\mathcal{E}(X)] = \tr X$ for any $X \in \mathcal{B}(\mathcal{H})$. A superoperator $\mathcal{E}$ is called unital if $\mathcal{E}(I_A) = I_B$, where $I_A$ is the identity operator in Hilbert space $\mathcal{H}_A$. A superoperator $\mathcal{E}$ is called positive if $\mathcal{E}(X) \geq 0$ for any $X \geq 0$, and $\mathcal{E}$ is called completely positive (CP) if $\text{id}_k \otimes \mathcal{E}$ is positive for any $k \geq 1$, where $\text{id}_k$ is the identity superoperator on $\mathcal{B}(\mathbb{C}^k)$.
A quantum channel is a physical procedure that describes the evolution $\mathcal{E}(\rho) = \rho'$, and the quantum channel $\mathcal{E}$ is a CPTP superoperator.

In order to describe the superoperator, here we use the Choi operator representation. For any superoperator $\mathcal{E}: \mathcal{B}(\mathcal{H}_A) \rightarrow \mathcal{B}(\mathcal{H}_B)$, its Choi operator $J(\mathcal{E}) \in \mathcal{B}(\mathcal{H}_A \otimes \mathcal{H}_B)$ is defined as
\begin{equation}
	J(\mathcal{E}) := \sum_{i, j = 0}^{d_A} \ketbra{i}{j}_A \otimes \mathcal{E}_{A \rightarrow B}\left(\ketbra{i}{j}_A\right).
\end{equation}
Conversely, the superoperator $\mathcal{E}$ can be reconstructed from its Choi operator $J(\mathcal{E})$ via the equation as follows,
\begin{equation}
	\mathcal{E}(X_A) = \tr_A \left[ J(\mathcal{E}) \left( X_A^T \otimes I_B \right) \right].
\end{equation}

We have summarized some important properties of superoperator $\mathcal{E}$ and the corresponding Choi operator $J(\mathcal{E})$ as the Lemma~\ref{lemma:superoperatorProperties}, which are useful for the following discussions. For more information, readers can refer the textbook by John Watrous~\cite{watrous2018theory}.

\begin{lemma}[Properties of superoperator and Choi representation]\label{lemma:superoperatorProperties}
For any superoperator $\mathcal{E}:\mathcal{B}(\mathcal{H}_A) \rightarrow \mathcal{B}(\mathcal{H}_B)$, we have the following properties:
\begin{itemize}
	\item $\mathcal{E}$ is completely positive iff $J(\mathcal{E})$ is positive.
	\item $\mathcal{E}$ is trace preserving iff $\tr_B J(\mathcal{E}) = I_A$.
	\item $\mathcal{E}$ is unital iff $\tr_A J(\mathcal{E}) = I_B$.
	\item $\mathcal{E}$ is completely positive iff $\mathcal{E}^\dagger$ is completely positive.
	\item $\mathcal{E}$ is trace preserving iff $\mathcal{E}^\dagger$ is unital.
\end{itemize}
\end{lemma}

\section{Results}

In this section, we will present our main results, including the calculation of one-shot coherence dilution $C_{\mathcal{O}}^{(1), \epsilon}$, as well as a comparison with a recent work published on \textit{Physical Review Letters}\cite{zhaoOneShotCoherenceDilution2018}.

\subsection{Analysis of one-shot coherence dilution}

The calculation of the fidelity between two density operators is hard to process in the applicatin of one-shot coherence dilution. Thus, we only consider pure state $\ket{\phi}$ other than a general mixed state $\rho$.
In this case, the fidelity follows $F(\Lambda(\Psi_m), \phi) = \inner{\Lambda(\Psi_m)}{\phi} =  \tr \Lambda(\Psi_m) \phi$.

Here we rewrite the one-shot coherence dilution Eq.~\eqref{def:dilution} as an optimization problem
\begin{equation} \label{eq:dilution-inequality}
\begin{aligned}
	C_{\mathcal{O}}^{(1), \epsilon} (\phi) = \log_2 \min 
    & \; m \in \mathbb{N} \\
	\text{s.t.} 
    & \; F(\Lambda(\Psi_m), \phi) \geq 1 - \epsilon \\
	& \; \Lambda \in \mathcal{O}
\end{aligned} \tag{$P1$}
\end{equation}

Firstly, we will show that the constraint $F(\Lambda(\Psi_m), \phi) \allowbreak \geq 1 - \epsilon$ is unnecessary, and we can have a stronger constraint $F(\Lambda(\Psi_m), \phi) = 1 - \epsilon$.

\begin{proposition} \label{prop:equality-1}
	The optimization problems \eqref{eq:dilution-inequality} and \eqref{eq:dilution-equality} have the same optimal solution.
    \begin{equation} \label{eq:dilution-equality}
    \begin{aligned}
		C_{\mathcal{O}}^{(1), \epsilon} (\phi) = \log_2 \min 
        & \; m \in \mathbb{N} \\
		\mathrm{s.t.} 
        & \; F(\Lambda(\Psi_m), \phi) = 1 - \epsilon \\
		& \; \Lambda \in \mathcal{O}
	\end{aligned} \tag{$P'1$}
    \end{equation}
\end{proposition}

\begin{proof}
	To prove this, one direction is that the feasible area of $F(\Lambda(\Psi_m), \phi) \geq 1 - \epsilon$ is larger than $F(\Lambda(\Psi_m), \phi) = 1 - \epsilon$. Therefore, the optimal solution  $p_1^\star \leq p_1'^\star$.
 
	Another direction is that if $F(\Lambda(\Psi_m), \phi) > 1 - \epsilon$, we can add a depolarizing channel $\mathcal{D}(X) = p \cdot \tr X \cdot \frac{I}{m} + (1 - p) X$ before $\Lambda$.
	\begin{align}
		& \inner{\Lambda (\mathcal{D} (\Psi_m))}{\phi} \nonumber\\
		& = \inner{\Psi_m}{\mathcal{D}(\Lambda^\dagger (\phi))} \\
		& = p \cdot \frac{1}{m} \cdot \tr \Lambda^\dagger (\phi) + (1 - p) \cdot \inner{\Psi_m}{\Lambda^\dagger (\phi)} \\
		& \in \left[\frac{\tr \Lambda^\dagger (\phi)}{m}, \inner{\Phi_m}{\Lambda^\dagger (\phi)} \right]
	\end{align}
	The first equality is because the adjoint of the depolarizing channel is itself. That is for any $X, Y \in L(\mathcal{H})$, we have
	\begin{align}
		\inner{\mathcal{D}(X)}{Y} 
		& = \inner{p \cdot \tr X \frac{I}{m} + (1-p) X}{Y} \\
		& = \frac{p}{m} \cdot  \tr X \cdot \tr Y + (1-p) \inner{X}{Y} \\
		& = \inner{X}{p \cdot \tr Y \frac{I}{m} + (1-p) Y} \\
		& = \inner{X}{\mathcal{D}(Y)}
	\end{align}
	Since $\mathcal{D} \in \text{SIO}$, $\Lambda' = \Lambda \circ \mathcal{D} \in \mathcal{O}$. Therefore, for any $\Lambda \in \mathcal{O}$ such that $F(\Lambda(\Psi_m), \phi) > 1 - \epsilon$, there exists $\Lambda' \in \mathcal{O}$ such that $F(\Lambda'(\Psi_m), \phi) = 1 - \epsilon$, which implies $p_1^\star \geq p_1'^\star$. In summary, $p_1^\star = p_1'^\star$. 
\end{proof}

Using the same technique as in paper \cite{fangProbabilisticDistillationQuantum2018,regulaOneShotCoherenceDistillation2018}, we can use the twirling channel $\mathcal{T}$ to simplify constraints. The twirling channel $\mathcal{T} : \mathcal{B}(\mathbb{C}^m) \rightarrow \mathcal{B}(\mathbb{C}^m)$ is defined as
\begin{equation}
	\mathcal{T}(\cdot) := \frac{1}{m!}\sum_{\pi \in S_m} U_\pi \cdot U_\pi^{-1}
\end{equation}
where $\pi$ is a permutation from the symmetric group $S_m$ and $U_\pi$ is the corresponding permutation matrix. Directly from the definition of $\mathcal{T}$, we know that the adjoint of $\mathcal{T}$ is itself, and it has the property that
\begin{equation}
	\mathcal{T} (\ketbra{i}{j}) = \begin{cases}
        I_m / m & , i = j \\
		\frac{1}{m-1} \left( \Psi_m - I_m / m \right) & , i \neq j
	\end{cases}
\end{equation}
Since $\Psi_m = \ketbra{\Psi_m}{\Psi_m}$ and $(I - \Psi_m) / (m-1)$ form an orthogonal basis for the image of $\mathcal{T}$, the image of $\mathcal{T}$ can be written as
\begin{equation}
	\mathcal{T} (\sigma) = a \cdot \Psi_m + b \cdot \frac{I - \Psi_m}{m-1}
\end{equation}
where $a = \inner{\Psi_m}{\mathcal{T}(\sigma)}$ and $b = \inner{I - \Psi_m}{\mathcal{T}(\sigma)}$.

Since $\Psi_m$ is invariant under $\mathcal{T}$, we can add a twirling channel before $\Lambda$.
\begin{align}
	\inner{\Lambda(\Psi_m)}{\phi}
	& = \inner{\Lambda \circ \mathcal{T} (\Psi_m)}{\phi} \\
	& = \inner{\Psi_m}{\mathcal{T} \circ \Lambda^\dagger (\phi)} \\
	& = \inner{\Psi_m}{a \cdot \Psi_m + b \cdot \frac{I - \Psi_m}{m-1}} \\
	& = a = 1 - \epsilon
\end{align}
Therefore, we denote
\begin{equation}
	\Psi_m^{\epsilon, b} = \mathcal{T} \circ \Lambda^\dagger (\phi) = (1 - \epsilon) \Psi_m + b \cdot \frac{I - \Psi_m}{m-1}
\end{equation}
where $b \geq 0$, because $\Lambda^\dagger$ is CP and unital.

Since $\mathcal{T} \in \text{SIO}$, $\Lambda \circ \mathcal{T} \in \mathcal{O}$. Now the problem \eqref{eq:dilution-equality} can be converted to
\begin{equation} \label{eq:dilution-p2}
\begin{aligned}
	C_{c, \mathcal{O}}^{(1), \epsilon} (\phi) = \log_2 \min 
        & \; m \in \mathbb{N} \\
	\text{s.t.} 
        & \; \Lambda^\dagger (\phi) = \Psi_m^{\epsilon, b} \\
	    & \; \Lambda \in \mathcal{O}
\end{aligned} \tag{$P2$}
\end{equation}

To process the constraint $\Lambda \in \mathcal{O}$, we consider the Choi matrix of $\Lambda^\dagger$, which is defined in Eq.~\eqref{ChoiMatrix}.
\begin{equation}\label{ChoiMatrix}
	J(\Lambda^\dagger) = \sum_{i, j} \ketbra{i}{j}_B \otimes \Lambda^\dagger (\ketbra{i}{j})_A
\end{equation}
Then, the Choi matrix of the operation $\mathcal{T} \circ \Lambda^\dagger$ is
\begin{align}
	\Tilde{J} & := J (\mathcal{T} \circ \Lambda^\dagger) \\
	& = \sum_{i,j} \ketbra{i}{j}_B \otimes \mathcal{T} \circ \Lambda^\dagger (\ketbra{i}{j})_A \\
	& = \frac{1}{m!} \sum_{\pi \in S_m} \sum_{i,j} \ketbra{i}{j}_B \otimes U_\pi \left[\Lambda^\dagger(\ketbra{i}{j})_A \right] U_\pi^{-1}  \\
	& = \frac{1}{m!} \sum_{\pi \in S_m} \left( I_B \otimes U_\pi \right) J(\Lambda^\dagger) \left( I_B \otimes U_\pi \right)^{-1} \\
	& = C^T \otimes \Psi_m + D^T \otimes \frac{I_m - \Psi_m}{m-1} \label{eq:choimatrixlast}
\end{align}
The Eq.~\eqref{eq:choimatrixlast} holds because that the operation $\mathcal{T}$ is applied on the system $A$.

Now we analyze the constraints on $\Tilde{J}$. Since $\Lambda \in \mathcal{O}$, $\Lambda$ is CPTP implies that $\Lambda^\dagger$ is CP and unital according to lemma~\ref{lemma:superoperatorProperties}. Here we consider the constraints in \eqref{eq:dilution-p2} one by one as follows.
\begin{itemize}
	\item $\mathcal{T} \circ \Lambda^\dagger$ is CP is equivalent to $\Tilde{J} \geq 0$.  Since $\Psi_m$ and $(I_m - \Psi_m) / (m-1)$ are positive and orthogonal, the constraint is further equivalent to $C \geq 0$ and $D \geq 0$.
	\item $\Lambda^\dagger$ is unital implies that $\mathcal{T}\circ \Lambda^\dagger$ is unital. Then, 
	\begin{align}
		I_A & = \Lambda^\dagger (I_B) \\
		& = \tr_B \Tilde{J} \left( I_B^T \otimes I_A \right) \\
		& = \tr_B \Tilde{J} \\
		& = \tr C \cdot \Psi_m + \tr D \cdot \frac{I - \Psi_m}{m-1} \\
		& = I = \Psi_m + I - \Psi_m
	\end{align}
	Thus, it is obvious to conclude the following relations, 
	\begin{equation}
		\begin{cases}
			\tr C = 1, \\
			\tr D = m - 1.
		\end{cases}
	\end{equation}
	\item The constraint $\Lambda^\dagger (\phi) = \Psi_m^{\epsilon, b}$ implies the following equations.
	\begin{align}
		& (1-\epsilon) \Psi_m + b \cdot \frac{I - \Psi_m}{m-1} \nonumber\\
		& = \tr_B \Tilde{J} \left( \phi^T \otimes I_A \right) \\
		& = \tr_B \left( C^T \phi^T \otimes \Psi_m + D^T \phi^T \otimes \frac{I - \Psi_m}{m-1} \right) \\
		& = \tr C \phi \cdot \Psi_m + \tr D \phi \cdot \frac{I - \Psi_m}{m-1}.
	\end{align}
	Therefore, we have
	\begin{equation}
		\begin{cases}
			\tr C \phi = 1 - \epsilon, \\
			\tr D \phi = b \geq 0.
		\end{cases}
	\end{equation}
	Since $D \geq 0$ and $\phi \geq 0$ imply that $\tr D \phi \geq 0$, we only need to keep the constraint $\tr C \phi = 1 - \epsilon$.\end{itemize}

Now we consider the one-shot dilution under MIO, that is $\mathcal{O} = \mathrm{MIO}$.
$\Lambda \in \text{MIO}$ means $\Lambda(\ketbra{i}{i}) = \Delta\circ \Lambda(\ketbra{i}{i})$ for $i = 0, \ldots, d-1$. 
There is a trick to represent $J(\mathcal{N}^\dagger)$ using $\mathcal{N} : \mathcal{B}(\mathcal{H}_A) \rightarrow \mathcal{B}(\mathcal{H}_B)$.
	\begin{equation}
		J(\mathcal{N}^\dagger) = 
		\sum_{i,j=1}^{d_A} \overline{\mathcal{N}(\ketbra{i}{j})}_B \otimes \ketbra{i}{j}_A
	\end{equation}
Using this trick, we can analyze $\Tilde{J}$.
	\begin{align}
		J(\mathcal{T} \circ \Lambda^\dagger)
		& = \sum_{i,j} \overline{\Lambda(\ketbra{i}{j})}_B \otimes \mathcal{T} (\ketbra{i}{j})_A \\
		& = \sum_{i} \overline{\Lambda(\ketbra{i}{i})} \otimes \frac{I}{m} \\
		& + \sum_{i \neq j} \overline{\Lambda(\ketbra{i}{j})} \otimes \frac{\Psi_m - \frac{I}{m}}{m-1} \nonumber\\
		& = \sum_{i} \overline{\Lambda(\ketbra{i}{i})} \otimes \frac{1}{m} \left(\Psi_m + I - \Psi_m\right) \\
		& + \sum_{i \neq j} \overline{\Lambda(\ketbra{i}{j})} \otimes \frac{1}{m} \left(\Psi_m - \frac{I - \Psi_m}{m-1}\right) \nonumber
	\end{align}
	Therefore, we can calculate $C$ and $D$ by partial trace.
	\begin{align}
		C^T & = \inner{\Psi_m}{\Tilde{J}}_A \\
		& = \frac{1}{m} \sum_i \overline{\Lambda(\ketbra{i}{i})}
		+ \frac{1}{m} \sum_{i \neq j} \overline{\Lambda(\ketbra{i}{j})} \\
		& = \frac{1}{m} \sum_{i,j} \overline{\Lambda(\ketbra{i}{j})}, \\
		D^T & = \inner{I - \Psi_m}{\Tilde{J}}_A \\
		& = \frac{m-1}{m} \sum_i \overline{\Lambda(\ketbra{i}{i})}
		- \frac{1}{m} \sum_{i \neq j} \overline{\Lambda(\ketbra{i}{j})} \\
		& = \sum_{i} \overline{\Lambda(\ketbra{i}{i})} - C^T.
	\end{align}
	$\Lambda(\ketbra{i})$ is diagonal implies $C^T + D^T$ is diagonal, that is $\Delta(C + D) = C + D$.

In summary, combined with all constraints, we have
\begin{equation}
\begin{aligned}
	C_{\text{MIO}}^{(1), \epsilon} (\phi) = \log_2 \min 
    & \; m \in \mathbb{N} \\
	\text{s.t.} 
    & \; C \geq 0, D \geq 0, \\
	& \; \tr C = 1 \\
	& \; \tr D = m-1 \\
	& \; \tr C \phi = 1 - \epsilon \\
	& \; \Delta(C + D) = C + D.
\end{aligned}
\end{equation}

Let $G \equiv C + D$. We have proved the following theorem.

\begin{theorem} \label{theorem:dilution-MIO}
    The one-shot coherence dilution of pure state under MIO is
    \begin{equation} \label{eq:dilution-pure-mio}
	\begin{aligned}
		C_{\mathrm{MIO}}^{(1), \epsilon} (\phi) = \log_2 \min 
        & \; \lceil \tr G \rceil \\
		\mathrm{s.t.} 
        & \;  0 \leq C \leq G \\
		& \; \tr C = 1 \\
		& \; \tr C \phi = 1 - \epsilon \\
		& \; \Delta(G) = G.
	\end{aligned} \tag{$P3$}
    \end{equation}
\end{theorem}

Now we consider the one-shot coherence dilution under DIO.
Based on $\Lambda \in \text{MIO}$, $\Lambda \in \text{DIO}$ adds a constraint $\Delta \circ \Lambda(\ketbra{i}{j}) = 0$ if $i \neq j$. According to the above representations of $C$ and $D$, we know that
\begin{equation}
	m C^T - (C^T + D^T) = \sum_{i \neq j} \overline{\Lambda(\ketbra{i}{j})}.
\end{equation}
Therefore, $\Delta[m C - (C + D)] = 0$.

In summary, we have the following theorem for the one-shot coherence dilution under DIO.

\begin{theorem} \label{theorem:dilution-DIO}
	The one-shot coherence dilution of pure state under DIO is
	\begin{equation} \label{eq:dilution-pure-dio}
	\begin{aligned}
        C_{\mathrm{DIO}}^{(1), \epsilon} (\phi) = \log_2 \min 
        & \; m \in \mathbb{N} \\
		\mathrm{s.t.} 
        & \;  0 \leq C \leq G \\
		& \; \tr C = 1 \\
		& \; \tr C \phi = 1 - \epsilon \\
		& \; G = \Delta(G) = m \Delta(C)
	\end{aligned} \tag{$P4$}.
    \end{equation}
\end{theorem}

Thm.~\ref{theorem:dilution-MIO} gives a SDP form of $C^{(1), \epsilon}_{\text{MIO}}$, but in Thm.~\ref{theorem:dilution-DIO} $C^{(1), \epsilon}_{\text{DIO}}$ has a quadratic constraint.

The definition of one-shot coherence dilution has an integer requirement. Since the rounding up function is troublesome, here we rewrite the aforementioned formulas without rounding up.
\begin{equation} \label{eq:dilution-pure-mio-without-rounding-up}
\begin{aligned}
		\Tilde{C}_{\mathrm{MIO}}^{(1), \epsilon} (\phi) = \log_2 \min 
        & \; \tr G \\
		\mathrm{s.t.} 
        & \;  0 \leq C \leq G \\
		& \; \tr C = 1 \\
		& \; \tr C \phi = 1 - \epsilon \\
		& \; \Delta(G) = G,
\end{aligned} \tag{$P'3$}
\end{equation}
and
\begin{equation} \label{eq:dilution-pure-dio-without-rounding-up}
\begin{aligned}
	\Tilde{C}_{\mathrm{DIO}}^{(1), \epsilon} (\phi) = \log_2 \min 
    & \; m \\
	\mathrm{s.t.} 
    & \;  0 \leq C \leq G \\
	& \; \tr C = 1 \\
	& \; \tr C \phi = 1 - \epsilon \\
	& \; G = \Delta(G) = m \Delta(C)
\end{aligned} \tag{$P'4$}
\end{equation}

\subsection{Dual form of optimization}

Note the fact that the dual form of an optimization problem is usually easier to solve and thereby having practical significance.
In this subsection, we calculate the dual forms of \eqref{eq:dilution-pure-mio-without-rounding-up} and \eqref{eq:dilution-pure-dio-without-rounding-up}. The skills for processing dual forms can be found in the popular textbook by Boyd and Vandenberghe \cite{boyd2004convex}.

Consider the Lagrange of \eqref{eq:dilution-pure-mio-without-rounding-up} 
\begin{align}
	L & = \inner{G}{I} - \inner{C}{\Lambda_1} + \inner{C - G}{\Lambda_2} \\
	& + a \left( \inner{C}{I} - 1 \right) \notag \\
	& + b \left(  \inner{C}{\phi} - (1 - \epsilon) \right) \notag \\
	& + \sum_{i \neq j} c_{i,j} \inner{G}{E_{i,j}} \notag \\
	& = \inner{G}{I - \Lambda_2 + \sum_{i \neq j} c_{i,j} E_{i,j}} \\
	& + \inner{C}{-\Lambda_1 + \Lambda_2 + a I + b \phi} \notag \\
	& - a - b(1 - \epsilon). \notag
\end{align}

Then, consider the dual function
\begin{align}
	g & = \inf_{C, G} L = - a - b(1 - \epsilon) \\
	& \text{if} \quad \begin{cases}
		I - \Lambda_2 + \sum_{i \neq j} c_{i,j} E_{i,j} = 0 \\
		-\Lambda_1 + \Lambda_2 + a I + b \phi = 0.
	\end{cases}
\end{align}

Therefore, the dual form of \eqref{eq:dilution-pure-mio-without-rounding-up}  is
\begin{equation} \label{eq:dual-MIO}
\begin{aligned}
	\max & \quad -a - b(1-\epsilon) \\
	\text{s.t.} & \quad \Lambda_1 \geq 0, \Lambda_2 \geq 0 \\
	& \quad I - \Lambda_2 + \sum_{i \neq j} c_{i,j} E_{i,j} = 0 \\
	& \quad -\Lambda_1 + \Lambda_2 + a I + b \phi = 0.
\end{aligned}
\end{equation}Since the dual form Eq.~\eqref{eq:dual-MIO} has a strict solution of $a = b = c = 0$, the strong duality holds according to the Slater's condition.

Although the \eqref{eq:dilution-pure-dio-without-rounding-up} is not a SDP, we can still calculate the dual form.
Rewrite the \eqref{eq:dilution-pure-dio-without-rounding-up} as
\begin{equation}
\begin{aligned}
	\min & \quad \inner{G}{I} \\
	\text{s.t.}
	& \quad -C \leq 0 \\
	& \quad C - G \leq 0 \\
	& \quad \inner{C}{I} - 1 = 0 \\
	& \quad \inner{C}{\phi} - (1-\epsilon) = 0 \\
	& \quad \inner{G}{E_{i,j}} = 0, \quad \forall i \neq j \\
	& \quad \inner{G - m C}{E_{i,i}} = 0, \quad \forall i
\end{aligned}
\end{equation}

Consider the Lagrange
\begin{align}
	L & = \inner{G}{I} - \inner{C}{\Lambda_1} + \inner{C - G}{\Lambda_2} \\
	& + a (\inner{C}{I} - 1) \notag \\
	& + b (\inner{C}{\phi} - (1 - \epsilon)) \notag \\
	& + \sum_{i \neq j} c_{i,j} \inner{G}{E_{i,j}} \notag \\
	& + \sum_i d_i \inner{G - m C}{E_{i,i}} \notag \\
	& = \inner{G}{I - \Lambda_2 + \sum_{i \neq j} c_{i,j} E_{i,j} + \sum_i d_i E_{i,i}} \\
	& + \inner{C}{-\Lambda_1 + \Lambda_2 + a I + b \phi - m \sum_i d_i E_{i,i}} \notag \\
	& - a - b(1 - \epsilon). \notag
\end{align}

Firstly, consider $\inf_G L$,
\begin{align}
	& \inf_G L = - a - b(1 - \epsilon) \\
    & + \inner{C}{-\Lambda_1 + (a+1) I + b \phi + \sum_{i \neq j} c_{i,j} E_{i,j} } \nonumber \\
    & + \inner{C}{(1-m) \sum_i d_i E_{i,i}} \nonumber \\
	\text{if} &  \quad I - \Lambda_2 + \sum_{i \neq j} c_{i,j} E_{i,j} + \sum_i d_i E_{i,i} = 0.
\end{align}

Then consider 
\begin{align}
	& \inf_{m, C} \inf_{G} L \nonumber \\
	& = \inf_{m,C} -a-b(1-\epsilon) + \inner{C}{R} + \inner{C}{(1-m) Q} \\
	& = \begin{cases}
		-\infty , & \quad \text{otherwise} \\
		-a-b(1-\epsilon), & \quad R = Q = 0
	\end{cases}
\end{align}

Therefore, the dual form of \eqref{eq:dilution-pure-dio-without-rounding-up} is
\begin{equation}
\begin{aligned}
	\max & \quad -a-b(1-\epsilon) \\
	\text{s.t.} 
	& \quad \sum_i d_i E_{i,i} = 0 \\
	& \quad \Lambda_2 = I + \sum_{i \neq j} c_{i,j} E_{i,j} + \sum_i d_i E_{i,i} \geq 0 \\
	& \quad \Lambda_1 = (a+1) I + b \phi + \sum_{i \neq j} c_{i,j} E_{i,j} \geq 0
\end{aligned}
\end{equation}
Rewrite it as
\begin{equation} \label{eq:dual-DIO}
\begin{aligned}
	\max & \quad a + b(1-\epsilon) \\
	\text{s.t.}
	& \quad \sum_{i \neq j} c_{i,j} E_{i,j} \leq I \\
	& \quad a I + b \phi + \sum_{i \neq j} c_{i,j} E_{i,j} \leq I.
\end{aligned}
\end{equation}
Eq.~\eqref{eq:dual-DIO} is just the same as Eq.~\eqref{eq:dual-MIO}. We find that \eqref{eq:dilution-pure-mio-without-rounding-up} and \eqref{eq:dilution-pure-dio-without-rounding-up} have the same dual form. We have proved that the strong duality holds for MIO, and according to section~\ref{sec:numerical-experiment} we know there is a gap between MIO and DIO. Therefore, the strong duality doesn't hold for DIO.

\subsection{Compared with previous work}

Zhao and his collaborators used the method of inequalities to provide the upper and lower bounds for $C_{\mathcal{O}}^{(1), \epsilon}$. Here we conclude their main results as the following theorem~\ref{theorem:zhao2018}.
\begin{theorem}[Zhao (2018) \cite{zhaoOneShotCoherenceDilution2018}] \label{theorem:zhao2018}
    For any state $\rho \in \mathbb{D}$ and $\epsilon \geq 0$, its one-shot coherence dilution under MIO and DIO are bounded by the following inequalities,
    \begin{align}
        C_{\max}^\epsilon (\rho) & \leq
        C_{\mathrm{MIO}}^{(1), \epsilon} (\rho) \leq
        C_{\max}^\epsilon (\rho) + 1, \\
        C_{\max, \Delta}^\epsilon (\rho) & \leq
        C_{\mathrm{DIO}}^{(1), \epsilon} (\rho) \leq
        C_{\max, \Delta}^\epsilon (\rho) + 1,
    \end{align}
    respectively.
\end{theorem}
The quantities presented in Thm.~\ref{theorem:zhao2018} are defined as follows:
\begin{definition}
    \begin{align}
        C_{\max} (\rho)
            & := \min_{\delta \in \mathcal{I}} D_{\max}(\rho \parallel \delta) \\
            & = \log_2 \min \left\lbrace \lambda \mid \delta \in \mathcal{I}, \rho \leq \lambda \delta \right\rbrace, \\
        C_{\max, \Delta} (\rho)
            & := D_{\max} (\rho \parallel \Delta(\rho)) \\
            & = \log_2 \min \left\lbrace \lambda \mid \rho \leq \lambda \Delta(\rho) \right\rbrace, \\
        C_{\max}^\epsilon (\rho)
            & := \min_{\rho' : F(\rho, \rho') \geq 1-\epsilon} C_{\max} (\rho'), \label{eq:ecmax} \\
        C_{\max, \Delta}^\epsilon (\rho)
            & := \min_{\rho' : F(\rho, \rho') \geq 1-\epsilon} C_{\max, \Delta} (\rho'). \label{eq:ecmaxd}
    \end{align}
\end{definition}
Eric Chitambar has proved that $C_{\max}$ is a coherence monotone under MIO and $C_{\max, \Delta}$ is a coherence monotone under DIO \cite{chitambarComparisonIncoherentOperations2016}.

Here we compare our results Thm.~\ref{theorem:dilution-MIO} and Thm.~\ref{theorem:dilution-DIO} with inequalities Thm.~\ref{theorem:zhao2018}.

Firstly, we rewrite $C_{\max}^\epsilon$ as an optimization problem.
\begin{equation}
\begin{aligned}
    C_{\max}^\epsilon (\phi)
        = \log_2 \min 
        & \; \lambda \\
        \text{s.t.} 
        & \; \rho \in \mathbb{D}, \delta \in \mathcal{I} \\
        & \; \tr (\phi \rho) \geq 1 - \epsilon \\
        & \; 0 \leq \rho \leq \lambda \delta .
\end{aligned}
\end{equation}
Like Prop.~\ref{prop:equality-1}, we can use an equality instead of the inequality as the following proposition.
\begin{proposition}
The quantity defined in Eq.~\eqref{eq:ecmax} can be strictly written as the following form,
\begin{equation} \label{eq:smoothed-c-max-pure}
    \begin{aligned}
        C_{\max}^\epsilon (\phi)
            = \log_2 \min 
            & \; \lambda \\
            \mathrm{s.t.} 
            & \; \rho \in \mathbb{D}, \delta \in \mathcal{I} \\
            & \; \tr (\phi \rho) = 1 - \epsilon \\
            & \; 0 \leq \rho \leq \lambda \delta. 
    \end{aligned}
\end{equation}
\end{proposition}
\begin{proof}
    Suppose the inequality is strict.
    Since $C_{\max}$ is monotone under MIO, then $C_{\max}^\epsilon \geq 0$ and $\lambda \geq 1$.
    If there exists $\rho^\star \in \mathbb{D}$, $\delta^\star \in \mathcal{I}$, $\lambda^\star \geq 1$ such that $F(\phi, \rho^\star) > 1-\epsilon$, $\rho^\star \leq \lambda^\star \delta^\star$ and
    the optimal solution is $C_{\max}^\epsilon (\phi) = \log_2 \lambda^\star$. We use the depolarizing channel to get 
    \begin{equation}
        \rho' = \mathcal{D}(\rho^\star) = p \frac{I}{m} + (1-p) \rho^\star.
    \end{equation}
    Consider the fidelity between $\rho'$ and $\phi$.
    \begin{align}
        F(\phi, \rho')
        & = \tr \left[ \phi \mathcal{D}(\rho^\star) \right] \\
        & = p \cdot \frac{1}{m} + (1-p) \cdot \tr \phi \rho^\star.
    \end{align}
    Since $\tr \phi \rho^\star > 1-\epsilon$, we can choose the proper $p$ such that $\tr \phi \rho' = 1-\epsilon$. Then consider the inequality.
    \begin{align}
        \rho'
        & = p \frac{I}{m} + (1-p) \rho^\star \\
        & \leq p \frac{I}{m} + (1-p) \lambda^\star \delta^\star \\
        & = \lambda' \delta'.
    \end{align}
    for some $\delta' \in \mathcal{I}$. And it's easy to check
    \begin{equation}
        \lambda' = \tr \left[p \frac{I}{m} + (1-p) \lambda^\star \delta^\star \right] = 1 \cdot p + \lambda^\star \cdot (1-p).
    \end{equation}
    Further, $\lambda^\star \geq 1$ implies that $\lambda' \leq \lambda^\star$. Therefore, we can always achieve the optimal solution by using $F(\phi, \rho) = 1-\epsilon$.
\end{proof}

Similarly, we can get a strict version of $C_{\max, \Delta}^\epsilon$, which is concluded in Proposition~\ref{pro:cmaxd}.
\begin{proposition}\label{pro:cmaxd}
The quantity defined in Eq.~\eqref{eq:ecmaxd} can be strictly written as the following form,
\begin{equation}
    \begin{aligned} \label{eq:smoothed-c-max-delta-pure}
        C_{\max, \Delta}^\epsilon (\phi)
        = \log_2 \min 
        & \; \lambda \\
        \mathrm{s.t.}
        & \; \rho \in \mathbb{D} \\
        & \; \tr \phi \rho = 1-\epsilon \\
        & \; 0 \leq \rho \leq \lambda \Delta(\rho).
    \end{aligned}
\end{equation}
    
\end{proposition}

Compared with \eqref{eq:dilution-pure-mio-without-rounding-up} and \eqref{eq:dilution-pure-dio-without-rounding-up}, it's easy to find that they are identical to Eq.~\eqref{eq:smoothed-c-max-pure} and Eq.~\eqref{eq:smoothed-c-max-delta-pure}, respectively. This observation is concluded as the following theorem.

\begin{theorem} \label{theorem:relationship-between-dilution-and-monotone}
    For any pure state $\phi \in \mathbb{D}$,
    \begin{align}
        \Tilde{C}_{\mathrm{MIO}}^{(1), \epsilon} (\phi) & = C_{\max}^\epsilon (\phi), \\
        \Tilde{C}_{\mathrm{DIO}}^{(1), \epsilon} (\phi) & = C_{\max, \Delta}^\epsilon (\phi).
    \end{align}
\end{theorem}
Thm.~\ref{theorem:relationship-between-dilution-and-monotone} gives a relationship between one-shot coherence dilution and coherence monotone, which shows that the lower bound in Thm.~\ref{theorem:zhao2018} is strict when considering pure states.

\section{Numerical Experiment} \label{sec:numerical-experiment}

This section lists some experiment results.

Firstly, we compare the one-shot coherence dilution under MIO and DIO. Though $C_{\text{DIO}}^{(1), \epsilon}$ is not a SDP, but we can still calculate its numerical value. Given a number $m$, the calculation of \eqref{eq:dilution-pure-dio} is a SDP, and obviously only when $m \geq m^\star$ the feasible area is not empty. Therefore, we can calculate $m^\star$ through a bisection method. As $m \in \mathbb{N}$, the result will be discrete. Therefore, we just discard the integer requirement and calculate the continuous version $\Tilde{C}_{\mathrm{MIO}}^{(1), \epsilon}$ in \eqref{eq:dilution-pure-mio-without-rounding-up} and $\Tilde{C}_{\mathrm{DIO}}^{(1), \epsilon}$ in \eqref{eq:dilution-pure-dio-without-rounding-up}.

The test dataset is generated by making superposition of $\ket{\Psi_d}$ and $\ket{0}$, which are the maximally coherent state and one incoherent state respectively.
\begin{equation}
    \ket{\phi} = \alpha \ket{\Psi_d} + \beta \ket{0}
\end{equation}
where $\alpha \in [0, 1]$, $\beta = -\alpha C + \sqrt{\alpha^2 (C^2 - 1) + 1}$ and we denote $C := \braket{\Psi_d}{0}$.
As $\alpha$ increases, $\ket{\phi}$ has more and more coherent components. Fig.~\ref{fig:test1} shows the one-shot coherence dilution rate under MIO and DIO, where $\epsilon = 0.01$.

\begin{figure}[htb]
	\centering
	\includegraphics[width = 0.45\textwidth]{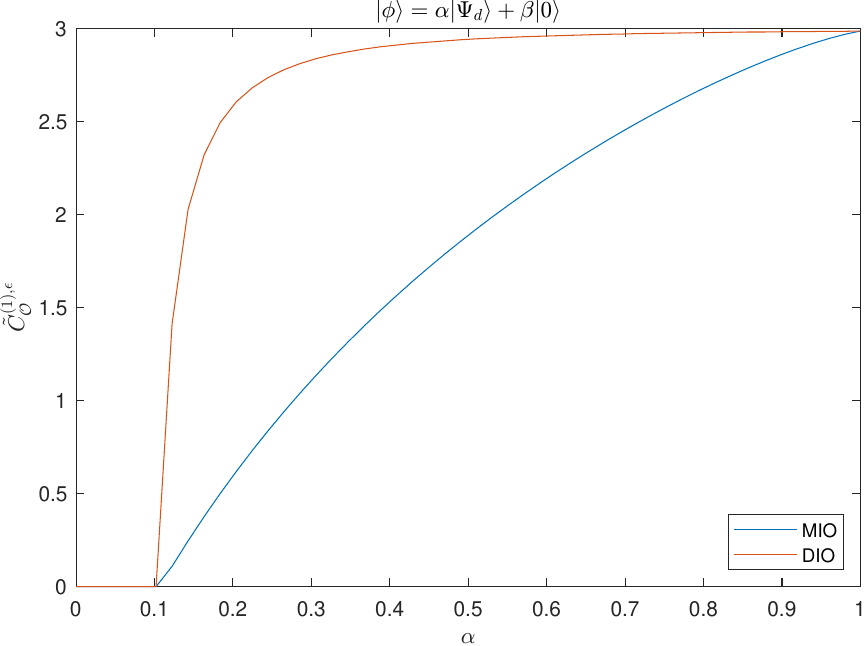}
	\caption{One-shot coherence dilution $\Tilde{C}_{\mathcal{O}}^{(1), \epsilon}$ under MIO and DIO. The state has the form $\ket{\phi} = \alpha \ket{\Psi_d} + \beta \ket{0}$. Here we choose $d = 8$ and $\epsilon = 0.01$.}
	\label{fig:test1}
\end{figure}

It can be clearly seen from the experiment's result that there is a gap between $\Tilde{C}_{\text{MIO}}^{(1), \epsilon}$ and $\Tilde{C}_{\text{DIO}}^{(1), \epsilon}$. And we know $\text{DIO} \subset \text{MIO}$ implies that $C_{ \text{MIO}} \leq C_{ \text{DIO}}$. This experiment shows that this inequality is strict.

We also consider how these two curves behave with different error tolerance $\epsilon$, which is depicted in Fig.~\ref{fig:test2}. As we can see, no matter what $\epsilon$ is, there is always a gap between $\Tilde{C}_{\text{MIO}}^{(1), \epsilon}$ and $\Tilde{C}_{\text{DIO}}^{(1), \epsilon}$ and two curves behave like Fig.~\ref{fig:test1}. The two curves separate at a certain point and finally converge at the same point. As $\epsilon$ increases, the separating point moves right and the converging point moves down. And the curves with higher tolerance $\epsilon$ are lower than those with lower tolerance. That means $\Tilde{C}_{\mathcal{O}}^{(1), \epsilon} > \Tilde{C}_{\mathcal{O}}^{(1), \epsilon'}$ when $\epsilon < \epsilon'$, which is absolutely right.
\begin{figure}[htb]
	\centering
	\includegraphics[width=0.45\textwidth]{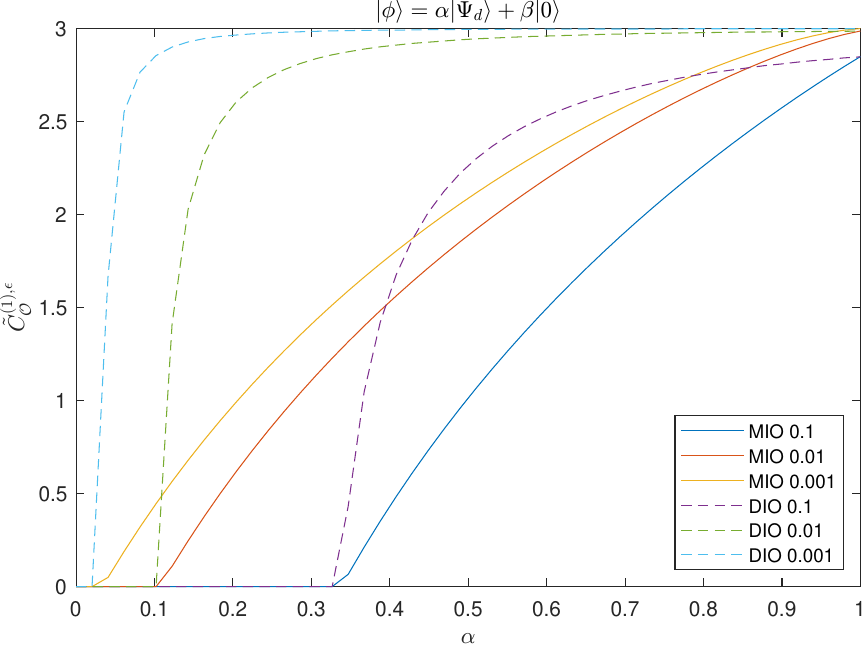}
	\caption{One-shot coherence dilution $\Tilde{C}_{\mathcal{O}}^{(1), \epsilon}$ under MIO and DIO with different error tolerance. The state has the form $\ket{\phi} = \alpha \ket{\Psi_d} + \beta \ket{0}$. Here we choose $d = 8$ and $\epsilon \in \{0.1, 0.01, 0.001\}$.}
	\label{fig:test2}
\end{figure}

\section{Conclusion}

In this paper, we have proposed a SDP form of the one-shot coherence dilution of pure states under MIO and a similar form under DIO. 
Although $C_{\text{DIO}}^{(1), \epsilon}$ is not a semidefinite programming, we can still use a bisection method to get its numerical value. 
Moreover, we compared our results with the inequalities in the previous work, and we have proved that the lower bound is strict for pure states.
Through numerical experiments, we clearly show that MIO and DIO have different power in one-shot coherence dilution, even though they have the same power in the asymptotic scenario and coherence distillation.

The coherence dilution of general states $\rho$ is hard to represent, since we still don't know how to treat the fidelity between two mixed states.
Our results are only applicable to pure states. It is still an open problem that whether the lower bounds of Thm.~\ref{theorem:zhao2018} are strict.
And we have already known the asymptotic and one-shot scenarios of coherence dilution, but the second order of dilution is still an open problem. We hope our numerical results will be of some help.

\begin{acknowledgments}
    We are with grateful respect to thank Professor Xing Wang and his group member for meaningful advises and discussions. This research was partially supported by the Innovation Program for Quantum Science and Technology (Grant No. 2021ZD0302900), National Natural Science Foundation of China (Grants No. 62002333).
\end{acknowledgments}


\bibliographystyle{unsrt}
\bibliography{ref}

\begin{thebibliography}{10}

\bibitem{PhysRevLett.111.250404}
Fernando G. S.~L. Brand\~ao, Micha\l{} Horodecki, Jonathan Oppenheim, Joseph~M. Renes, and Robert~W. Spekkens.
\newblock Resource theory of quantum states out of thermal equilibrium.
\newblock {\em Physical Review Letters}, 111:250404, December 2013.

\bibitem{sciadv.aaz4888}
Jianshu Cao, Richard~J. Cogdell, David~F. Coker, Hong-Guang Duan, Jürgen Hauer, Ulrich Kleinekathöfer, Thomas L.~C. Jansen, Tomáš Mančal, R.~J.~Dwayne Miller, Jennifer~P. Ogilvie, Valentyn~I. Prokhorenko, Thomas Renger, Howe-Siang Tan, Roel Tempelaar, Michael Thorwart, Erling Thyrhaug, Sebastian Westenhoff, and Donatas Zigmantas.
\newblock Quantum biology revisited.
\newblock {\em Science Advances}, 6(14):eaaz4888, April 2020.

\bibitem{giorda2017coherence}
Paolo Giorda and Michele Allegra.
\newblock Coherence in quantum estimation.
\newblock {\em Journal of Physics A: Mathematical and Theoretical}, 51(2):025302, December 2017.

\bibitem{PhysRevLett.116.150502}
Carmine Napoli, Thomas~R. Bromley, Marco Cianciaruso, Marco Piani, Nathaniel Johnston, and Gerardo Adesso.
\newblock Robustness of coherence: An operational and observable measure of quantum coherence.
\newblock {\em Physical Review Letters}, 116:150502, April 2016.

\bibitem{PhysRevA.93.042107}
Marco Piani, Marco Cianciaruso, Thomas~R. Bromley, Carmine Napoli, Nathaniel Johnston, and Gerardo Adesso.
\newblock Robustness of asymmetry and coherence of quantum states.
\newblock {\em Physical Review A}, 93:042107, April 2016.

\bibitem{PhysRevLett.116.240405}
A.~Streltsov, E.~Chitambar, S.~Rana, M.~N. Bera, A.~Winter, and M.~Lewenstein.
\newblock Entanglement and coherence in quantum state merging.
\newblock {\em Physical Review Letters}, 116:240405, June 2016.

\bibitem{PhysRevA.93.012111}
Mark Hillery.
\newblock Coherence as a resource in decision problems: The {{Deutsch-Jozsa}} algorithm and a variation.
\newblock {\em Physical Review A}, 93:012111, January 2016.

\bibitem{PhysRevA.95.032307}
Hai-Long Shi, Si-Yuan Liu, Xiao-Hui Wang, Wen-Li Yang, Zhan-Ying Yang, and Heng Fan.
\newblock Coherence depletion in the {{Grover}} quantum search algorithm.
\newblock {\em Physical Review A}, 95:032307, March 2017.

\bibitem{siddiqi2021engineering}
Irfan Siddiqi.
\newblock Engineering high-coherence superconducting qubits.
\newblock {\em Nature Reviews Materials}, 6(10):875--891, September 2021.

\bibitem{luo2023quantum}
L~Luo, M~Mootz, Jong-Hoon Kang, Chuankun Huang, Kitae Eom, JW~Lee, C~Vaswani, YG~Collantes, EE~Hellstrom, Ilias~E Perakis, et~al.
\newblock Quantum coherence tomography of light-controlled superconductivity.
\newblock {\em Nature Physics}, 19(2):201--209, December 2023.

\bibitem{aberg2006quantifying}
Johan Aberg.
\newblock Quantifying superposition.
\newblock {\em arXiv preprint quant-ph/0612146}, 2006.

\bibitem{baumgratzQuantifyingCoherence2014}
T.~Baumgratz, M.~Cramer, and M.~B. Plenio.
\newblock Quantifying coherence.
\newblock {\em Physical Review Letters}, 113(14):140401, September 2014.

\bibitem{chitambarQuantumResourceTheories2019}
Eric Chitambar and Gilad Gour.
\newblock Quantum resource theories.
\newblock {\em Reviews of Modern Physics}, 91(2):025001, April 2019.

\bibitem{PhysRevLett.117.030401}
Eric Chitambar and Gilad Gour.
\newblock Critical examination of incoherent operations and a physically consistent resource theory of quantum coherence.
\newblock {\em Physical Review Letters}, 117:030401, July 2016.

\bibitem{PhysRevA.94.052324}
Iman Marvian and Robert~W. Spekkens.
\newblock How to quantify coherence: Distinguishing speakable and unspeakable notions.
\newblock {\em Physical Review A}, 94:052324, November 2016.

\bibitem{winterOperationalResourceTheory2016}
Andreas Winter and Dong Yang.
\newblock Operational resource theory of coherence.
\newblock {\em Physical Review Letters}, 116(12):120404, March 2016.

\bibitem{yuanIntrinsicRandomnessMeasure2015}
Xiao Yuan, Hongyi Zhou, Zhu Cao, and Xiongfeng Ma.
\newblock Intrinsic randomness as a measure of quantum coherence.
\newblock {\em Physical Review A}, 92(2):022124, August 2015.

\bibitem{duConditionsCoherenceTransformations2015}
Shuanping Du, Zhaofang Bai, and Yu~Guo.
\newblock Conditions for coherence transformations under incoherent operations.
\newblock {\em Physical Review A}, 91(5):052120, May 2015.

\bibitem{zhuOperationalOnetooneMapping2017}
Huangjun Zhu, Zhihao Ma, Zhu Cao, Shao-Ming Fei, and Vlatko Vedral.
\newblock Operational one-to-one mapping between coherence and entanglement measures.
\newblock {\em Physical Review A}, 96(3):032316, September 2017.

\bibitem{regulaOneShotCoherenceDistillation2018}
Bartosz Regula, Kun Fang, Xin Wang, and Gerardo Adesso.
\newblock One-shot coherence distillation.
\newblock {\em Physical Review Letters}, 121(1):010401, July 2018.

\bibitem{zhaoOneShotCoherenceDilution2018}
Qi~Zhao, Yunchao Liu, Xiao Yuan, Eric Chitambar, and Xiongfeng Ma.
\newblock One-shot coherence dilution.
\newblock {\em Physical Review Letters}, 120(7):070403, February 2018.

\bibitem{hayashiFiniteBlockLength2021}
Masahito Hayashi, Kun Fang, and Kun Wang.
\newblock Finite block length analysis on quantum coherence distillation and incoherent randomness extraction.
\newblock {\em IEEE Transactions on Information Theory}, 67(6):3926--3944, June 2021.

\bibitem{watrous2018theory}
John Watrous.
\newblock {\em The theory of quantum information}.
\newblock Cambridge university press, 2018.

\bibitem{fangProbabilisticDistillationQuantum2018}
Kun Fang, Xin Wang, Ludovico Lami, Bartosz Regula, and Gerardo Adesso.
\newblock Probabilistic distillation of quantum coherence.
\newblock {\em Physical Review Letters}, 121(7):070404, August 2018.

\bibitem{boyd2004convex}
Stephen Boyd and Lieven Vandenberghe.
\newblock {\em Convex optimization}.
\newblock Cambridge university press, 2004.

\bibitem{chitambarComparisonIncoherentOperations2016}
Eric Chitambar and Gilad Gour.
\newblock Comparison of incoherent operations and measures of coherence.
\newblock {\em Physical Review A}, 94(5):052336, November 2016.

\end{thebibliography}

\end{document}